\documentclass[journal,twocolumn,Manuscript]{IEEEtran}

\ifCLASSINFOpdf
  
\else
  
\fi

\usepackage{amsmath}
\usepackage{amsthm}
\usepackage{cases}
\usepackage{amsfonts}
\usepackage{cite}
\usepackage{amssymb}
\usepackage{enumerate}
\usepackage{graphicx}
\usepackage{float}
\usepackage{caption}
\usepackage{verbatim}
\usepackage[cmintegrals]{newtxmath}
\usepackage[ruled,linesnumbered]{algorithm2e}
\usepackage{enumitem}
\usepackage{stfloats}

\newtheorem{theorem}{Theorem}
\newtheorem{lemma}{Lemma}
\newtheorem{definition}{Definition}
\newtheorem{corollary}{Corollary}
\newtheorem{remark}{Remark}

\newtheorem{example}{Example}
\begin{document}

\title{Distribution-Preserving Integrated Sensing and Communication with Secure Reconstruction}
%
%

%
\author{
  \IEEEauthorblockN{Yiqi Chen\IEEEauthorrefmark{1},
  Tobias Oechtering\IEEEauthorrefmark{2},
  Holger Boche\IEEEauthorrefmark{3},
  Mikael Skoglund\IEEEauthorrefmark{2},
  and Yuan Luo\IEEEauthorrefmark{1}}\\
\IEEEauthorblockA{\IEEEauthorrefmark{1}%
Shanghai Jiao Tong University,  
Shanghai,
China, 
\{chenyiqi,yuanluo\}@sjtu.edu.cn}\\
\IEEEauthorblockA{\IEEEauthorrefmark{2}%
KTH Royal Institute of Technology,  
100 44 Stockholm,
Sweden, 
\{oech,skoglund\}@kth.se}\\
\IEEEauthorblockA{\IEEEauthorrefmark{3}%
Technical University of Munich,  
80290 Munich, Germany, 
boche@tum.de}
\thanks{Yiqi Chen and Yuan Luo were supported in part by  National Key R\&D Program of China under Grant 2022YFA1005000, and  National Natural Science Foundation of China under Grant 62171279.}\\
 \thanks{Holger Boche was supported by the German Federal Ministry of Education and Research (BMBF) within the programme of ”Souverän. Digital. Vernetzt”, research HUB 6G-life, project identification number: 16KISK002, and by the BMBF Quantum Projects  QuaPhySI, Grant 16KIS1598K, and QUIET, Grant 16KISQ093.}
}


\maketitle
\thispagestyle{empty}
\begin{abstract}
Distribution-preserving integrated sensing and communication with secure reconstruction is investigated in this paper. In addition to the distortion constraint, we impose another constraint on the distance between the reconstructed sequence distribution and the original state distribution to force the system to preserve the statistical property of the channel states. An inner bound of the distribution-preserving capacity-distortion region is provided with some capacity region results under special cases. A numerical example demonstrates the tradeoff between the communication rate, reconstruction distortion and distribution preservation. Furthermore, we consider the case that the reconstructed sequence should be kept secret from an eavesdropper who also observes the channel output. An inner bound of the tradeoff region and a capacity-achieving special case are presented.
\end{abstract}


%
\IEEEpeerreviewmaketitle

\section{introduction}
Integrated sensing and communication (ISAC) is regarded as a promising feature of future wireless communication techniques such as 6G\cite{schwenteck20236g}.  It opens up completely new use cases and possibilities for users and operators for 6G, but it also poses special challenges for the design of communication systems. In particular, the misuse of ISAC could significantly increase the attack surface for security and privacy attacks on 6G \cite{fettweis20216g}\cite{fettweis20226g}. However, security and privacy are essential prerequisites for a trustworthy 6G system \cite{fettweis20216g}\cite{fettweis20226g}. Therefore, it is of fundamental importance to further develop the theoretical foundations of ISAC to prevent misuse of ISAC technology.

In \cite{sutivong2005channel}, the author addressed the rate-distortion tradeoff for decoder side sensing, with the channel states known at the encoder side noncausally. Further developments can be found in \cite{kim2008state} for state amplification and \cite{zhang2011joint} where the states are not available at the encoder. In \cite{ahmadipour2022information}, the sender is interested in not only the information transmission, but also the estimation of the channel states for discrete channels. Fundamental limits of the capacity-distortion tradeoff for several channel models were provided. Extensions on ISAC over Gaussian channels were provided in \cite{xiong2023fundamental}. The secure ISAC system was studied in \cite{gunlu2023secure}, where the feedback signal is used for both sensing and secret key generation to protect the transmitted message. ISAC with binary input and additive white Gaussian noise was studied in \cite{gunlu2023secure2}. Source coding with distribution-preservation can be found in \cite{wagner2022rate,chen2022rate,saldi2015output}

\begin{figure}
    \centering
    \includegraphics[scale=0.45]{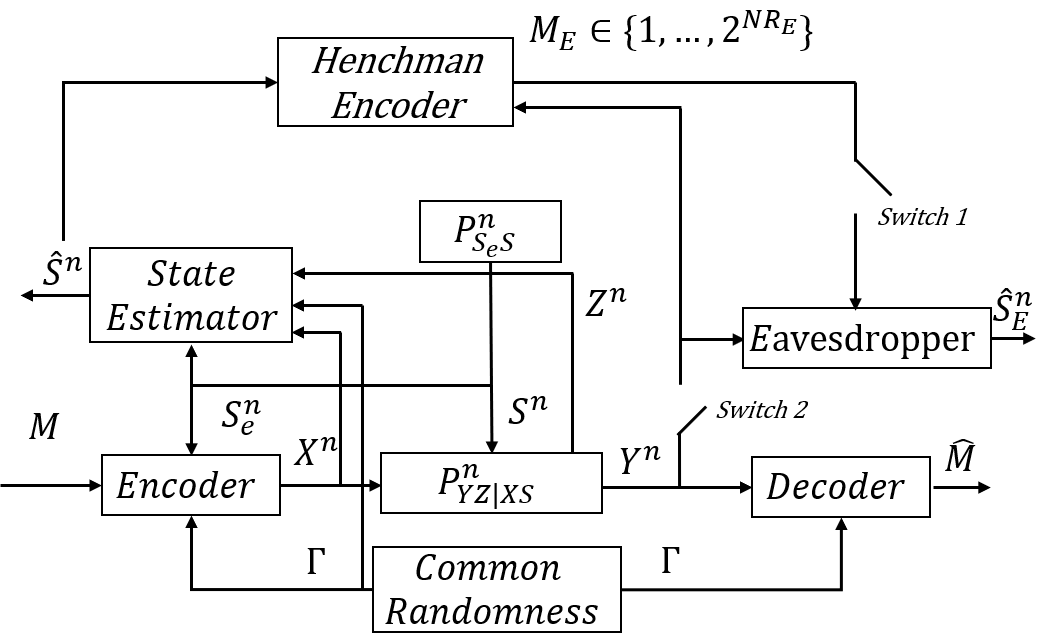}
    \caption{{\footnotesize The state estimator has all the information at the encoder side with an additional feedback sequence $z^n$. It tries to reconstruct the state sequence preserving the distribution. The henchman shares all the information at the state estimator and communicates to the eavesdropper with a limited rate.}}
    \label{fig: isac model}
\end{figure}
This paper considers a distribution-preserving ISAC system in that we not only reconstruct the state sequence within a given distortion constraint, but also require the distribution of the reconstructed sequence to be close to the original state distribution. This model is motivated by the fact that the ultimate goal of sensing is to learn the environment and design a better communication strategy for the system using the sensed state data. By preserving the state distribution, one can build a system for a specific environment and hence avoid the effect of the distribution mismatch problem. 
 To better understand the impact of ISAC on the security and privacy attack surface for future communication systems such as 6G, we further introduce a Henchman and Eavesdropper model into the ISAC framework and investigate the corresponding rate-distortion region.

\section{model and results}
Throughout this paper, random variable, sample value and its alphabet are denoted by capital, lowercase letters, and calligraphic letters, respectively, e.g. $X$, $x$, and $\mathcal{X}$. Symbols $X^n$ and $x^n$ represent random sequence and its sample value with length $n$. The distribution of a random variable $X$ is denoted by $P_X$ and the joint distribution of a pair of random variables $(X,Y)$ is denoted by $P_{XY}$. The expectation of a function of the random variable $X$ is written as $\mathbb{E}_X\left[ f(X) \right]$. 
Consider a communication system where the sender receives a feedback sequence after sending a codeword to the receiver through a state-dependent channel and then estimates the channel state based on the knowledge it has. Furthermore, we assume that the sender has a noisy observation of the channel state sequence in a noncausal manner, which represents some prior knowledge of the environment at the encoder side. Compared to the existing ISAC model\cite{ahmadipour2022information} that has a distortion constraint on the reconstructed sequence, the state estimator here is imposed an additional distribution constraint. In more detail, we want the distribution of the reconstructed sequence to be close to the original sequence under a given distance metric. 

From a more practical point of view, for any state-dependent system, which is governed by an i.i.d. generated state sequence $S^n$ and equipped with an ISAC system, the sensing result $\hat{S}^n$ is used for further processing, which is usually characterized by a bounded function $\kappa: \mathcal{S}^n \to \mathbb{R}$ such that $\max_{\hat{s}^n}|\kappa(\hat{s}^n)|\leq \Upsilon$, where $\Upsilon$ is a positive number and $\mathcal{S}$ is the common alphabet of $S$ and $\hat{S}$. We denote the distributions of $S^n$ and $\hat{S}^n$ by $P^n_S=\prod_{i=1}^n P_S$ and $P_{\hat{S}^n}$ for a moment. The processing functions $\kappa$ vary depending on different objectives of the system, with a consensus that the processing result of the estimated sequence should be close to that of the original sequence, regardless of what the function $\kappa$ is. We control this gap between the estimation and the original sequences by controlling

\begin{small}
    \begin{align}
    \label{neq: expectation gap}\sup_{\kappa:\; \max_{\hat{s}^n}|\kappa(\hat{s}^n)|\leq \Upsilon}\left|\mathbb{E}_{P^n_S}\left[ \kappa(S^n) \right] - \mathbb{E}_{P_{\hat{S}^n}}\left[ \kappa(\hat{S}^n) \right] \right|,
\end{align}
\end{small}
which can be further bounded by 
\begin{small}
    \begin{align*}
    &\sup_{\substack{\kappa:\; \max_{\hat{s}^n}|\kappa(\hat{s}^n)|\leq \Upsilon}}\left|\mathbb{E}_{P^n_S}\left[ \kappa(S^n) \right] - \mathbb{E}_{P_{\hat{S}^n}}\left[ \kappa(\hat{S}^n) \right] \right| \\
    &\overset{(a)}{=} \max_{\substack{\kappa:\; \max_{\hat{s}^n}|\kappa(\hat{s}^n)|\leq \Upsilon}}\left|\mathbb{E}_{P^n_S}\left[ \kappa(S^n) \right] - \mathbb{E}_{P_{\hat{S}^n}}\left[ \kappa(\hat{S}^n) \right] \right|\\
    &\leq \Upsilon \sum_{s^n\in\mathcal{S}^n} \left| P^n_S(s^n)-P_{\hat{S}^n}(s^n)  \right|  
\end{align*}
\end{small}
where $(a)$ is by the fact that the set of functions $\kappa$ is a compact set. This implies that with a reconstructed distribution that is sufficiently close to the original one, we can control the expectation of the processing result defined in  \eqref{neq: expectation gap} no matter which processing function is used.
We further consider the case where there exists a henchman in the state estimator and an eavesdropper who both observe the channel output. The henchman sends messages to the eavesdropper with a limited rate, and the goal of the henchman is to reconstruct the state as well. The model is illustrated in Fig. \ref{fig: isac model}.
In the following, we consider the model without secure constraint, which is the case that Switches 1 and 2 are turned off.
A code used for such a communication system is defined as follows.
\begin{definition}
    A common randomness (CR)-assisted code $(n,R,R_c)$ for distribution preserving ISAC with noisy channel state information (NCSI) consists of
    \begin{itemize}
        \item A message set $\mathcal{M}=[1:2^{nR}]$,
        \item Common randomness available both at the encoder and decoder $\Gamma\in[1:2^{nR_c}]$,
        \item An encoder $f:\mathcal{M}\times\mathcal{S}^n_e \times \Gamma \to \mathcal{X}^n$,
        \item A decoder $g:\mathcal{Y}^n\times \Gamma \to \mathcal{M}$,
        \item A state estimator $h:\Gamma\times\mathcal{X}^n\times\mathcal{S}^n_e\times\mathcal{Z}^n\to\hat{\mathcal{S}}^n$.
    \end{itemize}
\end{definition}
Throughout this paper, we use total variation as the distance metric between two distributions and denote it by $||P-Q||$ for distributions $P$ and $Q$ on the same alphabet.
\begin{definition}
    Given non-negative real numbers $(D,\Delta)$, a pair $(R,R_c)$ is said to be $(D,\Delta)-$achievable if for any $\epsilon>0$ and sufficiently large $n$, there exists an $(n,R,R_c)$ code such that
        \begin{small}
            \begin{equation}\label{def: DP achievable}
            \begin{split}
                Pr\{M\neq \hat{M}\} \leq \epsilon,\;\mathbb{E}\left[ d(S^n,\hat{S}^n) \right] \leq D,|| P_{\hat{S}^n} - P_S^n || \leq \Delta,
            \end{split}
        \end{equation}
        \end{small}
where $P_S$ is the distribution of the channel state.
\end{definition}
Given the state source distribution $P_{S_eS }$ with marginal distribution $P_S=\psi$ and channel $P_{YZ|XS}$, let $\mathcal{P}(U,X,\hat{S})$ be the set of the joint distribution $P_{S_eSUXYZ\hat{S}}$ such that
\begin{small}
    \begin{align*}
    &P_{S_eS UXYZ\hat{S}}(s_e,s ,u,x,y,z,\hat{s})=P_{S_eS }(s_e,s )P_{U|S_e}(u|s_e)\\
    &\quad\quad\quad\quad P_{X|US_e}(x|u,s_e)P_{YZ|XS}(y,z|x,s)P_{\hat{S}|XS_eZ}(x,s_e,z).
\end{align*}
\end{small}
Define a set of distributions $\mathcal{P} = \{P_{S_eS UXYZ\hat{S}}\in\mathcal{P}(U,X,\hat{S}): P_{S}=P_{\hat{S}}=P_S\}$.
The following theorem gives an inner bound on the $(D,0)-$capacity-distortion region, where the state distribution is perfectly preserved in the reconstruction. 
\begin{theorem}{(Inner bound)}\label{the: capacity perception distortion with cr and noisy CSI}
    An inner bound of the CR-assisted distribution-preserving capacity-distortion  region  is
\begin{small}
    \begin{equation*}
    R(D,0) = \left\{
    \begin{aligned}
        (R, R_c)\in\mathbb{R}^2: &R\geq 0, R_c \geq 0, \exists P_{S_eS UXYZ\hat{S}} \in \mathcal{P},\\
        &R \leq I(U;Y )-I(U;S_e),\\
        &R_c \geq \max\{I(U;\hat{S}) - I(U;Y ),0\},\\
        &D \geq \mathbb{E}[d(S,\hat{S})].
    \end{aligned}
    \right.
\end{equation*}  
\end{small}
\end{theorem}
The first constraint is the reliable communication constraint which has a similar form as the Gel'fand-Pinsker coding rate, and the second constraint is the lower bound on the common randomness we need to preserve the distribution of the state source. The total randomness rate we need to preserve the distribution is $I(U;\hat{S})$, and the reduction $I(U;Y )$ comes from the selection of messages and the stochastic encoder.

\emph{Sketch of the proof: }Fix the underlying distribution $P_{S_eS UXYZ\hat{S}}=P_{S_eS }P_{U|S_e}P_{X|US_e}P_{YZ|XS}P_{\hat{S}|XZS_e}$ and define non-negative real numbers $R,R_c,R_p$ satisfying
\begin{small}
    \begin{align*}
    &R + R_p \leq I(U;Y ),\;\;R_p \geq I(U;S_e),R + R_p + R_c \geq I(U;\hat{S}).
\end{align*}
\end{small}

\emph{Codebook Generation: } For each message $m\in[1:2^{nR}]$, generate a subcodebook $\mathcal{C}(m)=\{u^n(m,i,\gamma):i\in[1:2^{nR_p}],\gamma\in[1:2^{nR_c}]\}$. The whole codebook is denoted by $\mathcal{C}=\cup_{m\in\mathcal{M}}\mathcal{C}(m)$.

For a given codebook $\mathcal{C}$ and underlying distribution, define distributions
\begin{small}
\begin{align}
    &Q_{M I \Gamma S^n_e S^n   X^n Y^n Z^n \hat{S}^n}=\frac{1}{2^{n(R+R_p+R_c)}} P_{S_eS XYZ\hat{S}|U}^n \notag \\
     \label{def: idealized distribution}&= \frac{1}{2^{n(R+R_p+R_c)}}Q_{S^n_e |M I \Gamma } Q_{S^n   X^n Y^n Z^n \hat{S}^n|M I \Gamma S^n_e },\\
    \label{def: induced distribution}&\bar{Q}_{M I \Gamma S^n_e S^n   X^n Y^n Z^n \hat{S}^n}=\frac{P_{S_e}^n P_{I|M\Gamma S^n_e} Q_{S^n   X^n Y^n Z^n \hat{S}^n|M I \Gamma S^n_e}}{2^{n(R+R_c)}},
\end{align}
\end{small}
where $Q_{S^n_e |M I \Gamma }(s^n_e|m,i,\gamma)= P^n_{S_e|U}(s^n_e|u^n(m,i,\gamma)),$
\begin{small}
    \begin{align*}
    &P_{I|M\Gamma S^n_e}(i|m,\gamma,s^n_e)=\frac{P^n_{S_e|U}(s^n_e|u^n(m,i,\gamma))}{\sum_{i}P^n_{S_e|U}(s^n_e|u^n(m,i,\gamma))},
\end{align*}
\end{small}
\noindent and $Q_{ S^n   X^n Y^n Z^n \hat{S}^n|M I \Gamma S^n_e} =P^n_{S |S_e}P^n_{X|US_e}P^n_{YZ|XS}P^n_{\hat{S}|XS_eZ}.$
Note that
\begin{small}
    \begin{align*}
    Q_{I|M\Gamma S^n_e}(i|m,\gamma,s^n_e)=P_{I|M\Gamma S^n_e}(i|m,\gamma,s^n_e).
\end{align*}
\end{small}
We call the distribution $Q$ as the idealized distribution\cite{goldfeld2019wiretap} since it is the distribution when we select codewords from the codebook uniformly at random. The distribution $\bar{Q}$ is the induced distribution since it is induced by the distribution $Q_{I|M\Gamma S^n_e}$, which is used as our likelihood encoder in the following coding scheme.

\emph{Encoding: } To transmit the message $m$ with an informed encoder having access to the common randomness $\gamma$ and noisy side information $s^n_e$, the encoder uses a likelihood encoder\cite{song2016likelihood,goldfeld2019wiretap} and picks the index randomly according to $Q_{I|M\Gamma S^n_e}(i|m,\gamma,s^n_e)$.
The codeword $x^n$ is then generated by
\begin{align*}
    P_{X|US_e}^n(x^n|u^n(m,i,\gamma),s^n_e).
\end{align*}

\emph{Decoding: } The decoder observes the channel output $y^n$ and has access to the common randomness $\gamma$. Both are used to look for a unique $(\hat{m},\hat{i})$ such that
$(u^n(\hat{m},\hat{i},\gamma),y^n)\in T^n_{P_{U Y},\delta}.$

\emph{State Estimation: } The state estimator observes the feedback $z^n$. It generates a sequence $\hat{s}^n$ according to $P^n_{\hat{S}|XS_eZ}(\hat{s}^n|x^n,s^n_e,z^n) $.

By the soft-covering lemma, with a sufficiently large bin size, the underlying distribution $P$ is close to the idealized distribution $Q$, which is close to the true distribution $\bar{Q}$ induced by the encoder/decoder pair when the size of the codebook is sufficiently large. By the triangle inequality of the total variation, we have $||\bar{Q}-P||\leq \epsilon$ with $\epsilon$ which can be arbitrarily small with a sufficiently large block length $n$. To achieve perfect preservation, i.e. $\epsilon=0$, we use an optimal transport argument to construct a new sequence having the desired distribution with an additional distortion that decays exponentially fast with block length $n$. The decoding error is first analyzed under the distribution $Q$ and then is upper bounded by the total variation between $Q$ and the true distribution $\bar{Q}$. The region in Theorem \ref{the: capacity perception distortion with cr and noisy CSI} is derived using the Fourier-Motzkin elimination of the following region.

\begin{small}
    \begin{equation}
     R'(D,0) = \left\{
    \begin{aligned}\label{eq: equivalent region}
        &(R,R_p, R_c)\in\mathbb{R}^3: R\geq 0, R_p \geq 0, R_c \geq 0,\\
        &\exists P_{S_eS UXYZ\hat{S}} \in \mathcal{P},\\
        &R + R_p \leq I(U;Y ),\\
        &R_p \geq I(U;S_e)\\
        &R + R_p + R_c \geq I(U;\hat{S}),\\
        &D \geq \mathbb{E}[d(S,\hat{S})].
    \end{aligned}
    \right.
\end{equation}
\end{small}
  The detailed proof is provided in Appendix \ref{app: proof of thm noisy csi}.
    \begin{remark}\label{rem: equivalent region remark}
    The constraint $R_p$ is the randomness that comes from a likelihood encoder. However, it is not all the randomness at the encoder side. Note that the channel input sequence $X^n$ is generated by distribution $P_{X|US_e}$ instead of a fixed function $x(u,s_e)$ as possible for the Gel'fand-Pinsker coding problem due to the additional randomness constraint related to $I(U;\hat{S})$.

    The constraints on $R_p$ and $R+R_P+R_c$ in \eqref{eq: equivalent region} imply that when $I(U;S_e) \geq I(U;\hat{S})$ we can always set $R_c=0$ regardless of what value $R$ is. However, as shown by Theorem \ref{the: capacity perception distortion with cr and noisy CSI}, a weaker condition $I(U;Y )\geq I(U;\hat{S})$ is sufficient for zero common randomness rate. In the following, we show that this condition does not change with the Fourier–Motzkin Elimination. It is sufficient to show that when $I(U;S_e) < I(U;\hat{S})\leq I(U;Y )$, the corner points in Theorem \ref{the: capacity perception distortion with cr and noisy CSI} fall into the projection of the region $R'(D,0)$ on $R$ and $R_c$, which means that we can always find a $R_p$ such that $(R,R_p,R_c=0)\in R'(D,0)$. We start with point $(0,\max\{I(U;\hat{S}) - I(U;Y ),0\})=(0,0)$. In this case, we set $R_p = I(U;Y )$ and the constraints in $R'(D,0)$ are all satisfied. For the point $(I(U;Y )-I(U;S_e),0)\neq(0,0)$, we set $R_p = I(U;S_e)$ and it follows that $R+R_p+R_c=I(U;Y )\geq I(U;\hat{S})$.
\end{remark}

\subsection{No Common Randomness}
The main difference between our problem and distribution-preserving lossy compression is that in the lossy compression problem, the encoder and the decoder communicate through a noiseless channel with unlimited capacity. Hence, one can remove the common randomness and compensate for it with a sufficiently large compression rate. However, in our problem, the communication rate is limited by the channel capacity. Hence, when common randomness is not available in the communication system, the distribution-preserving reconstruction is not always available. Define 
\begin{small}
    \begin{align*}
    \mathcal{P}_{NCR} &= \{P_{S_eS UXYZ\hat{S}}\in\mathcal{P}(U,X,\hat{S}): P_{S}=P_{\hat{S}},  I(U;\hat{S})\leq I(U;Y )\}.
\end{align*}
\end{small}

\begin{corollary}\label{coro: no common randomness inner bound}
    If $\mathcal{P}_{NCR}\neq \emptyset$, an inner bound of the distribution-preserving capacity-distortion region without common randomness is
    \begin{small}
        \begin{equation*}
        R_{NCR}(D,0) = \left\{
        \begin{aligned}
            (R,0)\in\mathbb{R}^2: &R\geq 0, \exists P_{S_eS UXYZ\hat{S}} \in \mathcal{P}_{NCR},\\
            &R  \leq I(U;Y ) - I(U;S_e),\\
            &D \geq \mathbb{E}[d(S,\hat{S})].
        \end{aligned}
        \right.
    \end{equation*} 
    \end{small}
\end{corollary}
The communication rate $R_{NCR}(D)$ is in general smaller than that in the region $R(D,0)$ in Theorem \ref{the: capacity perception distortion with cr and noisy CSI} due to the additional constraint $I(U;\hat{S})<I(U;Y )$ on the input distribution. When there is no common randomness available, all the randomness comes from message selection and stochastic encoder, whose rates are limited by the channel capacity. If the rate of reliable communication is not sufficient for the soft-covering lemma, then the distribution preservation is impossible.

The region in Corollary \ref{coro: no common randomness inner bound} is in general not tight. However, if the state $S=(S_e,S_e' )$, and the communication model is a deterministic main channel such that $P_{YZ|XS}=\mathbb{I}\{Y=y(X,S_e)\}P_{Z|XS_eS_e' }$, where $y(X,S_e)$ is some function $\mathcal{X}\times\mathcal{S}_e \to \mathcal{Y}$, then we have the following capacity region.
\begin{corollary}\label{coro: no common randomness capacity}
    For deterministic channels, the distribution-preserving capacity-distortion tradeoff region without common randomness is
    \begin{small}
        \begin{equation*}
        C_{NCR}(D,0) = \left\{
        \begin{aligned}
            &(R,0)\in\mathbb{R}^2: R\geq 0, \exists P_{XS_e YZ\hat{S}} \in \mathcal{P}_{NCR}^{DE},\\
            &R  \leq H(Y|S_e),\\
            &D\geq \mathbb{E}[d(S,\hat{S})],
        \end{aligned}
        \right.
    \end{equation*} 
    \end{small}
    where $\mathcal{P}_{NCR}^{DE} = \{P_{XS_eS_e' YZ\hat{S}}=P_{S_eS_e' }P_{X|S_e}\mathbb{I}\{Y=y(X,S_e)\}P_{Z|XS_e S_e'}P_{\hat{S}|XS_eZ}:  P_{S}=P_{\hat{S}}\}.$
\end{corollary}
\begin{proof}
    The achievability is proved by setting $U=Y$ in Corollary \ref{coro: no common randomness inner bound}, which is valid by the deterministic main channel property. For this choice we have
    \begin{align*}
        I(U;Y )=I(Y;Y )=H(Y) \geq I(U;\hat{S})=I(Y;\hat{S}).
    \end{align*}   
Hence, the condition $I(U;\hat{S})\leq I(U;Y )$ in $\mathcal{P}_{NCR}$ always holds so that no further common randomness is needed. The converse follows by using the Fanos inequality and a standard time-sharing argument.
\end{proof}

\subsection{Deterministic Encoder}
In this subsection, we consider the case where the encoder is not allowed to randomize. Note that the randomness at the encoder consists of two parts: the likelihood encoder and the distribution $P_{X|US_e}$. When the encoder is not allowed to randomize, the likelihood encoder becomes deterministic so that the distribution $P_{X|US_e}$ becomes a deterministic function. Define a new set $\widetilde{\mathcal{P}}(U,X,\hat{S})$ of distributions $P_{S_eS UXYZ\hat{S}}$ such that
\begin{align*}
    &P_{S_eS UXYZ\hat{S}}(s_e,s ,u,x,y,z,\hat{s})=P_{S_eS }(s_e,s )P_{U}(u)x(u,s_e)\\
    &\quad\quad\quad\quad P_{YZ|XS}(y,z|x,s)P_{\hat{S}|XS_eZ}(\hat{s}|x,s_e,z)
\end{align*}
Further, define $ \mathcal{P}_{DE} = \{P_{S_eS UXYZ\hat{S}}\in\widetilde{\mathcal{P}}(U,X,\hat{S}): P_{S}=P_{\hat{S}}\}.$
\begin{corollary}\label{the: cpd with cr and causal csi}
    An inner bound of the CR-assisted distribution-preserving capacity-distortion region with  a deterministic encoder is
\begin{small}
    \begin{equation*}
    C(D,0) = \left\{
    \begin{aligned}
        (R,R_c)\in\mathbb{R}^2: &R\geq 0,R_c \geq 0,\exists P_{S_eS UXYZ\hat{S}} \in \mathcal{P}_{DE},\\
        &R \leq I(U;Y ),\\
        &R + R_c \geq I(U;\hat{S}),\\
        &D \geq \mathbb{E}[d(S,\hat{S})].
    \end{aligned}
    \right.
\end{equation*} 
\end{small}   
\end{corollary}
The proof follows by setting $R_p=0$ in the region \eqref{eq: equivalent region}, which requires $I(U;S_e)=0$. This implies that the auxiliary random variable $U$ is independent of $S_e$ and the encoder uses the channel state information in a causal manner. One can further consider a more strict reconstruction constraint by defining the output distribution constraint as $P_{S^n_e\hat{S}^n}=\prod_{i=1}^n P_{S_e S}.$
Redefine the input distribution set as
\begin{small}
    \begin{align*}
    \mathcal{P}_{DE}' = \{P_{S_eS UXYZ\hat{S}}\in\widetilde{\mathcal{P}}(U,X,\hat{S}): P_{S_eS}=P_{S_e\hat{S}}\}.
\end{align*}
\end{small}
Then we have the following capacity region.
\begin{theorem}{(Capacity)}\label{the: cpd with cr and causal csi and strict constraint}
    For the deterministic encoder, the CR-assisted distribution-preserving capacity-distortion  region with causal CSI is
\begin{small}
    \begin{equation*}
    C(D,0) = \left\{
    \begin{aligned}
        &(R,R_c)\in\mathbb{R}^2: R\geq 0,R_c \geq 0,\exists P_{S_eS UXYZ\hat{S}} \in \mathcal{P}_{DE}',\\
        &R \leq I(U;Y ),\\
        &R + R_c \geq I(U;\hat{S},S_e),\\
        &D \geq \mathbb{E}[d(S,\hat{S})].
    \end{aligned}
    \right.
\end{equation*}    
\end{small}
\end{theorem}
The achievability proof follows by setting $R_p=0$ and updating the constraints $R + R_c \geq I(U;\hat{S},S_e)$ due to the more strict constraint. To prove the converse, the bound on $R$ follows using the converse proof for the channel with causal CSI\cite[Chapter 7]{el2011network}. The bound on $R+R_c$ follows by the inequality $n(R + R_c) \geq H(M,\Gamma) \geq I(M,\Gamma;S^n_e,\hat{S}^n)$ and the fact that $(S_{e,i},\hat{S}_i),i=1,2,\dots,n$ are i.i.d. random variable pairs. The details are provided in Appendix \ref{app: proof of theorem de and strict constraint}. 

\begin{remark}
    The regions in this subsection are included in the inner bound in Theorem \ref{the: capacity perception distortion with cr and noisy CSI} since when randomization is not allowed at the encoder side, although the reconstructed distribution can be preserved, the communication rate is reduced due to the usage of the channel state information. The result shows that the private randomness at the encoder side helps both the transmission and reconstruction.    
\end{remark}

\subsection{Boundary Points: $C(0,\Delta)$ and $C(D_{min},\Delta)$, $ \Delta \geq 0$}
In general, due to the existence of channel noise, zero distortion is impossible for ISAC problem, unless the state estimator can reproduce the exact state based on the input and feedback. In this case, not only the distortion is minimized, we also achieve perfect distribution preservation. To see this, let $h: \mathcal{X}^n\times\mathcal{Z}^n\to\mathcal{S}^n$ be a reconstruction function that can reproduce the states based on $X^n$ and $Z^n$. It follows that
\begin{small}
    \begin{align*}
    P_{\hat{S}^n}(\hat{s}^n)&=\sum_{m}\frac{1}{|\mathcal{M}|}\sum_{s^n}P^n_{S}(s^n)\sum_{z^n}P_{Z|XS}^n(z^n|x^n(m),s^n)\\
    &\quad\quad\quad\quad\quad\quad\quad \cdot P_{\hat{S}^n|X^nZ^n}(\hat{s}^n|x^n(m),z^n)\\
    &\overset{(a)}{=}\sum_{m}\frac{1}{|\mathcal{M}|}P^n_S(\hat{s}^n)\sum_{z^n}P_{Z|XS}^n(z^n|x^n(m),\hat{s}^n)\\
    &\quad\quad\quad\quad\quad\quad\quad \cdot\mathbb{I}\{\hat{s}^n=h(x^n(m),z^n)\}\\
    &=P^n_S(\hat{s}^n)=\prod_{i=1}^{n}P_S(\hat{s}_i),
\end{align*}
\end{small}
where $(a)$ follows by the fact that $\mathbb{I}\{\hat{s}^n=h(x^n(m),z^n)\}=0$ when $s^n\neq\hat{s}^n$. However, in general, a low distortion does not necessarily imply a well-preserved reconstructed distribution. Consider the following example in which the channel model is borrowed from \cite{ahmadipour2022information}.
\begin{example}{Binary Channel with Multiplicative Bernoulli Source:}
    Consider a channel $Y=SX$ with binary alphabets $\mathcal{X}=\mathcal{S}=\mathcal{Y}=\{0,1\}$ where the state $S \sim Bernoulli(q),q=\frac{1}{4}$. The feedback is perfect $Z=Y$ and the distortion measure is Hamming distortion $d(s,\hat{s})=s \oplus \hat{s}$.
We first show that the set $\mathcal{P}(D,0)$ is not empty for some $D>0$. Select the input distribution $X\sim Bernoulli(\frac{3}{4})$ and the reconstruction distribution $P_{\hat{S}|XZ}$ such that the transition matrix satisfies
    \begin{gather}
        \begin{bmatrix}\label{def: example reconstruction matrix}
            q & 1-q \\
            a & 1-a \\
            \frac{2q}{3(1-q)} & 1-\frac{2q}{3(1-q)} \\
            \frac{1}{3} & \frac{2}{3}
        \end{bmatrix}
    \end{gather}
with arbitrary $a\in[0,1]$. It follows that with the selection of input and reconstruction distributions, the underlying distribution of $\hat{S}$ satisfies $P_{\hat{S}}=P_{S}$. With the help of common randomness, one can always preserve the state distribution for the reconstruction. 
However, this reconstruction distribution is not the best estimator to minimize the distortion. As shown in \cite{ahmadipour2022information}, the best estimator is $\hat{s}^*(x,z)=z$ if $x=1$ and $\hat{s}^*(x,z)=0$ if $x=0$.
The underlying distribution of $\hat{S}$ in this case is
\begin{small}
    \begin{align*}
    P_{\hat{S}}(0)=P_X(0) + P_{X}(1)P_S(0),
\end{align*}
\end{small}
which is equal to $P_{S}$ when we choose input distribution such that $P_X(1)=1$. Then we have $Z=Y=XS=S$ and by the definition of $\hat{s}^*(x,z)$ the estimator always outputs $z=s$. However, the communication rate is 0 in this case.

Using the distribution $P_{\hat{S}|XZ}$ with the transition matrix \eqref{def: example reconstruction matrix} and Theorem \ref{the: capacity perception distortion with cr and noisy CSI} with $U=X,S_e=\emptyset$, we achieve points $(R,R_c)$ such that $P_{\hat{S}^n}=P_S^n$. Of course, in this case, since we are not using the best estimator in the sense of distortion, and hence the distortion $D$ is larger than $D_{min}$ that is achieved by $\hat{s}^*(x,z)$. However, with the input distribution $X\sim Bernoulli(\frac{3}{4})$, we also achieve a positive communication rate.
\end{example}

\section{Gaussian Example}
In this subsection, we consider the additive Gaussian channel
\begin{align*}
    Y=X+S+N,
\end{align*}
where $S\sim \mathcal{N}(0,\sigma_S^2),N\sim \mathcal{N}(0,\sigma_N)$ are independent Gaussian random variables with $\sigma_S^2=2,\sigma_N^2=1.$ The input constraint is $\frac{1}{n}\sum_{i=1}^n \mathbb{E}[X^2]\leq 2.$ The feedback signal is perfect feedback $Z=Y$ and the covariance matrix of the state $S$ and side information $S_e$ at the encoder side satisfies
\begin{align*}
    \left[
    \begin{matrix}
        \sigma_S^2 & \rho \sigma_S \sigma_{S_e}\\
        \rho \sigma_S \sigma_{S_e} & \sigma_{S_e}^2
    \end{matrix}
    \right]
\end{align*}
where $\sigma^2_{S_e}=3$ and $\rho\in[-1,1]$ is the correlation between $S$ and $S_e$.

We choose the auxiliary random variable $U=X+\alpha S_e$ for some $\alpha \in \mathbb{R}^+$ and the estimator 
\begin{align*}
    \hat{S} = a(X+\alpha S_e +Z) + W 
\end{align*}
where $W\sim \mathcal{N}(0,\frac{1}{2})$ is an independent Gaussian random variable. 

Now we can rewrite Theorem \ref{the: capacity perception distortion with cr and noisy CSI} as \eqref{region: gaussian example} at the bottom of the current page.
\begin{figure*}[b]
    \par\noindent\rule{\textwidth}{0.5pt}
    \begin{align}\label{region: gaussian example}
        R(D,0)=\left\{  
            \begin{aligned}
                &(\alpha,\rho)\in \mathcal{P} \\
                &R \leq \frac{1}{2}\log \frac{\sigma_X^2\sigma_Y^2}{(\sigma_X^2 + \alpha^2 \sigma_{S_e}^2)\sigma_Y^2-(\sigma_X^2+\alpha \rho \sigma_S\sigma_{S_e})^2},\\ 
                &R_c \geq \frac{1}{2}\log \frac{\sigma_S^2 \left( (\sigma_X^2 + \alpha^2 \sigma_{S_e}^2)\sigma_Y^2-(\sigma_X^2+\alpha \rho \sigma_S\sigma_{S_e})^2 \right) }{\sigma_Y^2 \left( (\sigma_X^2+\alpha^2\sigma_{S_e})\sigma_S^2 - (2a\sigma_X^2+a\alpha^2\sigma_{S_e}^2+a\alpha \rho \sigma_S\sigma_{S_e})^2 \right)}
            \end{aligned}
        \right.
    \end{align}
    where
    \begin{align*}
        &\mathcal{P}=\{(\alpha,\rho)\in\mathbb{R}^+ \times [-1,1]:\\
        &\quad\quad\quad\quad\quad\quad\quad 4a^2\sigma_X^2 + a^2\alpha^2\sigma_{S_e}^2 +(a^2-1)\sigma_S^2+a^2\sigma_N^2+2\alpha a^2(\rho\sigma_S\sigma_{S_e})+\sigma_W^2 =0,\\
        &\quad\quad\quad\quad\quad\quad\quad (2-2a)\sigma_S^2-2\alpha a^2 (\rho \sigma_S\sigma_{S_e}) \leq D            \}
    \end{align*}
\end{figure*}
\begin{figure}
    \centering
    \includegraphics[scale=0.5]{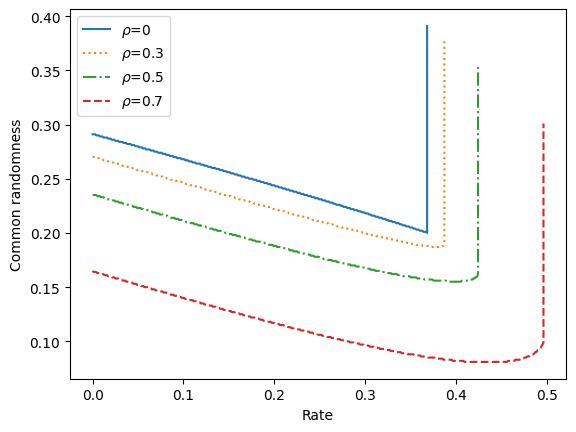}
    \caption{Rate-Common randomness region with different correlation values under distortion constraint $D=3$}
    \label{fig: gaussian example}
\end{figure}
We provide in Fig. \ref{fig: gaussian example} the rate-common randomness region with different correlation values of $S$ and $S_e$ under distortion constraint $D=3$. The figure shows that a higher correlation value $\rho$ implies a higher communication rate and a lower common randomness rate, which is reasonable since a high correlation implies the side information $S_e$ at the sender side is highly related to the channel state $S$, and hence is helpful in both communication and state estimation. Also note that when $\rho=0$, the Gaussian random variables $S_e$ and $S$ are independent of each other, and the sender can not benefit from observing the side information. The right end of the solid line is the case when $\alpha=0$ (which is not the case with different values of $\rho$). The sender does not use the side information in both the communication and estimation which is implied by Fig. \ref{fig: gaussian example} to be optimal in this case since the correlation $\rho=0$.

\section{Secure ISAC}


In this section, we consider the distribution-preservation ISAC with private reconstruction requirement. Suppose there exists a henchman at the state estimator side and an eavesdropper who both observe the channel output $Y^n$. However, the eavesdropper does not have access to the common randomness. The henchman observes the reconstructed sequence and the sequences at the encoder side, it transmits a rate-limited message to the eavesdropper to help the eavesdropper reconstruct the state sequence.

There is an additional distortion constraint that the distortion between reconstructed sequences at the state estimator side and the receiver side is lower bounded by a given number $D_E$. Let $\mathcal{C}$ be the communication codebook used between the sender and the receiver. The behavior of the henchman and the eavesdropper is defined as follows.

\begin{definition}
    A henchman in an ISAC system observes the behavior of the encoder. Its operation consists of 
    \begin{itemize}
        \item a message set $\mathcal{M}_E=[1:2^{nR_E}]$,
        \item an encoder $f_E: \hat{\mathcal{S}}^n \times \mathcal{Y}^n \times \Gamma\times \mathcal{C} \to \mathcal{M}_E$,
        \item a decoder $g_E: \mathcal{Y}^n\times\mathcal{M}_E\times \mathcal{C} \to \hat{\mathcal{S}}_E^n$. 
    \end{itemize}

\end{definition}
\begin{definition}
    Given real numbers $D\geq 0, D_E \geq 0, P\geq 0$, a CR-assisted code is $(D,D_E,P)-$ achievable if for any $\epsilon>0$ and sufficiently large $n$, there exists an $(n,R,R_c)$ code such that
    \begin{small}
        \begin{align*}
        &Pr\{M\neq \hat{M}\} \leq \epsilon, \;|| P_{\hat{S}^n} - P_S^n || \leq \Delta\\
        &\mathbb{E}\left[ d(S^n,\hat{S}^n) \right] \leq D,\;\mathbb{E}\left[ d(\hat{S}^n,\hat{S}^n_E) \right] \geq D_E.
    \end{align*}
    \end{small}
\end{definition}

The main result for this section is as follows.

\begin{theorem}{(Inner bound)}\label{the: cpddr with cr and noisy CSI and CSI}
    An inner bound of the secure CR-assisted distribution-preserving capacity-distortion  region  is
\begin{small}
    \begin{equation*}
    \begin{split}
        &C(D,D_E,0) = \\
    &\left\{
    \begin{aligned}
        &(R, R_c)\in\mathbb{R}^2: R\geq 0, R_c \geq 0, \exists P_{S_eS UXYZ\hat{S}} \in \mathcal{P},\\
        &R \leq I(U;Y )-I(U;S_e),\\
        &R_c \geq \max\{I(U;\hat{S}) - I(U;Y),0\},\\
        &D \geq \mathbb{E}[d(S,\hat{S})],\\
        &D_E \leq  \left\{
            \begin{aligned}
                &D(R_E,P_{\hat{S}Y})\quad\quad\quad\quad\quad\quad\quad\quad\quad\quad\;\;\;\;\;\;\; \text{if $R_E<R_c,$}\\
                &\min\{D(R_E,P_{\hat{S}Y}), D(R_E-R_c,P_{U\hat{S}Y})\}\quad\text{if $R_E\geq R_c.$}
            \end{aligned}
        \right.
    \end{aligned}
    \right.
    \end{split}
\end{equation*}
\end{small}
where $D(r,P_{U\hat{S}Y})$ is the rate-distortion function such that $r$ is the compression rate and $P_{U\hat{S}Y}$ is the joint distribution such that $\hat{S}$ is the source and $U$ and $Y$ are side information available at both sides. We set $D(r,P_{U\hat{S}Y})=\infty$ if $r<0$.
\end{theorem}
The eavesdropper observes the channel output $Y^n$ but does not have access to the common randomness. When $R_E > R_c$, the henchman can always use $R_c$ bits to describe the common randomness, and then, with high probability, the eavesdropper can decode the message and obtain the codeword $U^n$. In this case, distortion $D(R_E-R_c,P_{U\hat{S}Y})$, which is the distortion-rate function under joint distribution $P_{{U\hat{S}Y}}$ with $U$ and $Y$ being common side information, it is achievable to the eavesdropper. When $R_E<R_c$, the eavesdropper ignores the codeword, and uses a lossy compression code with the rate $R_E$ and $Y$ being the common side information at both sides.

\emph{Sketch of the proof: } To prove the achievability, the sender still uses the random binning coding provided in Theorem \ref{the: capacity perception distortion with cr and noisy CSI}. For the henchman-eavesdropper, the system is equivalent to the case that a codeword $U^n(M,I,\Gamma)$ is selected uniformly at random and produces two sequences, $\hat{S}^n$ and $Y^n$, through a discrete memoryless channel $P^n_{Y\hat{S}|U}$. The henchman observes $(\hat{S}^n,Y^n)$ and describes $\hat{S}^n$ with a limited rate, and the eavesdropper decoder observes $Y^n$. The system is an extended version of the model in \cite{schieler2016henchman} with additional side information at both sides. The optimal strategy for the henchman-eavesdropper is to use a deterministic decoder $s^n_E(M_E,Y^n)$ and an encoder $f_E$ that maps each pair of ($\hat{S}^n,Y^n)$ to an interim codeword $V^n\in \mathcal{C}_V$ that minimizes the expectation of the distortion and set $\hat{S}^n_E=V^n$. To show the distortion achieved at the eavesdropper, it is sufficient to prove $Pr_{\textbf{C}}\left\{\max_{f_E,\mathcal{C}_V}  Pr\{d(\hat{S}^n,\hat{s}^n_E(M_E,Y^n)) \leq D\} > \tau_n \right\}\to 0$ for sufficiently large $n$ and small $\tau_n.$ The subscript $\textbf{C}$ represents the random codebook used by the sender and legitimate receiver. Let $m_E(v^n)$ be the encoding result when $v^n\in\mathcal{C}_V$ is selected. It can be shown that
\begin{small}
    \begin{align*}
    &\max_{f_E,\mathcal{C}_V}  Pr\{d(\hat{S}^n,\hat{s}^n_E(M_E,Y^n))\leq D\}\\
    &\leq 2^{nR_E}\max_{v^n\in\mathcal{V}^n}Pr\{d(\hat{S}^n,\hat{s}^n_E(m_E(v^n),Y^n))\leq D\}
\end{align*}
\end{small}
by considering the `worst case' that the best $v^n$ for the henchman-eavesdropper is in the codebook and the factor $2^{nR_E}$ is the result of the union bound and the fact that given side information $Y^n$ there are at most $2^{nR_E}$ possible decoding results for the eavesdropper decoder.
The problem reduces to determining the compression rate given fixed $v^n\in\mathcal{V}^n$ and distortion constraint $D$ when side information $Y^n$ is available at the decoder side.
The details are provided in Appendix \ref{app: proof of henchman theorem}.

The following capacity-distortion region result can be obtained when the main channel is deterministic.
\begin{corollary}{(Capacity region)}\label{coro: deterministic cpddr with cr and noisy CSI and CSI}
    The CR-assisted secure distribution-preserving capacity-distortion region of the deterministic channel  is
\begin{small}
    \begin{equation*}
    C(D,D_E,0) = \left\{
    \begin{aligned}
        (R, 0)\in\mathbb{R}^2: &R\geq 0, \exists P_{S_eS UXYZ\hat{S}} \in \mathcal{P}_{DE},\\
        &R \leq H(Y|S_e),\\
        &D \geq \mathbb{E}[d(S,\hat{S})],\\
        &D_E \leq D(R_E,P_{Y\hat{S}}).
    \end{aligned}
    \right.
\end{equation*}
\end{small}
\end{corollary}
The achievability proof follows by setting $U=Y$ in Theorem \ref{the: cpddr with cr and noisy CSI and CSI}, which is valid by the deterministic property of the main channel. For the converse, the bound of $R$ and $D$ remain the same as in Corollary \ref{coro: no common randomness capacity}. To show the bound on $D_E$, we use an extended version of the type covering lemma (\cite[Lemma 9.1]{csiszar2011information}), which is presented in \cite[Lemma 4]{schieler2016henchman}. For each joint type of $(\hat{S}^n,Y^n)$ there exists a code by the type covering lemma, and the number of types is $(n+1)^{|\mathcal{S}||\mathcal{Y}|}$, which can be compressed with a negligible rate. The details are provided in Appendix \ref{app: converse proof of corollary henchman deterministic channel}.

\newpage

\bibliographystyle{ieeetr} 
\bibliography{ref}

\clearpage
\appendices
\section{coding scheme for theorem \ref{the: capacity perception distortion with cr and noisy CSI}}\label{app: proof of thm noisy csi}
In this section, we prove the achievability of Theorem \ref{the: capacity perception distortion with cr and noisy CSI}.

Before giving the decoding error and distortion analysis, we first state two properties of total variation as follows.

\begin{lemma}{\cite[Property 1]{song2016likelihood}}\label{lem: property of tv}
    Given alphabet $\mathcal{X}$ and distributions $P$ and $Q$ on $\mathcal{X}$, the following statements hold:
    \begin{enumerate}[label=(\alph*)]
        \item   Let $\epsilon>0$ and $f(x)$ be a function in a bounded range with width $b\in\mathbb{R}^+$. It follows that
        \begin{align*}
            ||P-Q|| < \epsilon \Rightarrow |\mathbb{E}_P\left[ f(X)  \right] - \mathbb{E}_Q\left[ f(X)  \right]| \leq b \epsilon.
        \end{align*} \label{def: tv property 1}
        \item \label{def: tv property 2} Let $P_XP_{Y|X}$ and $Q_XP_{Y|X}$ be joint distributions on $\mathcal{X}\times\mathcal{Y}$. It follows that
        \begin{align*}
            ||P_XP_{Y|X} - Q_XP_{Y|X}|| = ||P_X - Q_X||.
        \end{align*}
    \end{enumerate}
\end{lemma}

We first analyze the decoding error induced by the distribution $Q$, then by the property of the total variation lemma \ref{lem: property of tv}\ref{def: tv property 1}, the error probability induced by another close distribution $\bar{Q}$ can also be bounded. Define the following error events and suppose $(M,I,\Gamma)=(1,1,\gamma)$
\begin{align*}
    &\mathcal{E}_1=\{(U^n(1,1,\gamma),Y^n)\notin T^n_{P_{U Y},\delta}\},\\
    &\mathcal{E}_2 = \{(U^n(m',i',\gamma),Y^n)\in T^n_{P_{U Y},\delta} \;\text{for some $(m',i')\neq (1,1)$}\}.
\end{align*}

By the standard joint typicality argument for Gel'fand-Pinsker coding, under distribution $Q$ the decoding error probability is bounded by
\begin{align*}
    Pr_{Q}\{M\neq\hat{M}\}\to 0
\end{align*}
if $R+R_p\leq I(U;Y )$. The remaining analysis relies on the following soft-covering lemma.
\begin{lemma}\label{lem: soft covering}
    Let $(U,V)\in\mathcal{U}\times\mathcal{V}$ be a pair of random variables with joint distribution $P_{UV}$. Let $\textbf{C}=\{U^n(i):i\in[1:2^{nR}]\}$ be a random collection of sequences $U^n(i)$, each i.i.d. generated by $P_U^n$. We denote the sample value of $\textbf{C}$ by $\mathcal{C}$. Let $P_V$ be the underlying marginal distribution of $V$ induced by the joint distribution $P_{UV}$ and define the output distribution induced by random codebook $\textbf{C}$ as
    \begin{align*}
        Q_{V^n}(v^n)= \frac{1}{|\text{C}|}\sum_{i=1}^{2^{nR}}P_{V^n|U^n}(v^n|U^n(i)),\;\;\text{for any $v^n\in\mathcal{V}^n$},
    \end{align*}
where $P_{V^n|U^n}=\prod_{i=1}^n P_{V|U}$. If $R>I(U;V)$, we have
\begin{align*}
    \mathbb{E}_{\textbf{C}}\left[ || Q_{V^n} - P_V^n ||_{TV}\right] \leq \frac{3}{2}exp\{-\kappa n\}
\end{align*}
 for some $\kappa>0$.
\end{lemma}

Now invoking Lemma \ref{lem: soft covering}, we have
\begin{align*}
\mathbb{E}_{\textbf{C}(m,\gamma)}\left[ || Q_{S^n_e|M=m,\Gamma=\gamma} -P^n_{S_e} ||  \right] \leq \frac{3}{2}exp\{-an\}
\end{align*}
for some $a>0$ if $R_p > I(U;S_e)$, and hence,
\begin{align*}
    \mathbb{E}_{\textbf{C}}\left[ || Q_{M\Gamma S^n_e} -\frac{1}{2^{n(R+R_c)}}P^n_{S_e} ||  \right] \leq \frac{3}{2}\exp\{-an\},
\end{align*}
which gives us 
\begin{align}
    &\mathbb{E}_{\textbf{C}}\left[ || Q_{M I \Gamma S^n_e S^n   X^n Y^n Z^n \hat{S}^n} -\bar{Q}_{M I \Gamma S^n_e S^n   X^n Y^n Z^n \hat{S}^n} ||  \right] \notag \\
    &=\mathbb{E}_{\textbf{C}}\left[ || Q_{M  \Gamma S^n_e} -\bar{Q}_{M  \Gamma S^n_e } ||  \right] \notag \\
    \label{eq: idealized distribution approximation}&= \mathbb{E}_{\textbf{C}}\left[ || Q_{M\Gamma S^n_e} -\frac{1}{2^{n(R+R_c)}}P^n_{S_e} ||  \right] \leq \frac{3}{2}\exp\{-an\}
\end{align}
by Lemma \ref{lem: property of tv}\ref{def: tv property 2}. For a given random codebook $\mathbf{C}$, define the function $g_{\mathbf{C}}(m,\hat{m})=\mathbb{I}\{m\neq \hat{m}\}$. Now follow the same argument as \cite[Eq. (41)-(44)]{goldfeld2019wiretap} with Lemma \ref{lem: property of tv}\ref{def: tv property 1} we have 
\begin{align}\label{neq: error probability}
    Pr_{\bar{Q}}\{M\neq \hat{M}\} \leq \epsilon\cdot \frac{3}{2}\exp\{-an\}\to 0  \;\;\text{as}\;\; n\to \infty
\end{align}

On the other hand, by setting $R+R_p+R_c > I(U;\hat{S})$ and applying Lemma \ref{lem: soft covering}, we have
\begin{align*}
    &\mathbb{E}_{\textbf{C}}\left[ || Q_{\hat{S}^n} - P_{\hat{S}}^n ||  \right]\\
    &=\mathbb{E}_{\textbf{C}}\left[ || Q_{\hat{S}^n} - P_{S}^n ||  \right] \leq \frac{3}{2}\exp \{-bn\}
\end{align*}
for some $b>0$. By the triangle inequality property of the total variation \cite[Property 1(c)]{song2016likelihood}, we have
\begin{align}
    &\mathbb{E}_{\textbf{C}}\left[ || \bar{Q}_{\hat{S}^n} - P_{S}^n ||  \right]\notag \\
     &\leq \mathbb{E}_{\textbf{C}}\left[ || \bar{Q}_{\hat{S}^n} - Q_{\hat{S}^n} ||  \right] + \mathbb{E}_{\textbf{C}}\left[ || Q_{\hat{S}^n} - P_{S}^n ||  \right] \notag\\
     \label{neq: target distribution approximation}&\leq  \frac{3}{2}\exp\{-an\} +  \frac{3}{2}\exp\{-bn\} :=  \frac{3}{2}\exp\{-cn\}
\end{align}
for some $c>0$. For the reconstruction distortion part, we first consider the distortion under distribution $Q$ and define additional error events
\begin{align*}
    &\mathcal{E}_3 = \{(U^n(1,1,\gamma),S^n_e,S^n)\notin T^n_{P_{US_eS },\delta}\},\\
    &\mathcal{E}_4 = \{(X^n,S^n_e,Z^n,\hat{S}^n)\notin T^n_{P_{XS_eZ\hat{S}},\delta}\}
\end{align*}
We say a distortion error $\mathcal{E}_d$ occurs if at least one of the events $\mathcal{E}_1,\mathcal{E}_2,\mathcal{E}_3$ and $\mathcal{E}_4$ happens. By the law of large numbers, it follows that $Pr\{\mathcal{E}_3\}\to 0$ and $Pr\{\mathcal{E}_4\}\to 0$ as $n\to \infty$ and then $P_{d,e} := Pr\{\mathcal{E}_d\}\to 0$ follows. The expectation of the distortion (taken over random codebook, random states and random channel noise) under distribution $Q$ is bounded by
\begin{align*}
    &\mathbb{E}\left[ d(S^n,\hat{S}^n)  \right] \\
    &= (1-P_{d,e})\mathbb{E}\left[ d(S^n,\hat{S}^n) | \mathcal{E}_d^c \right] + P_{d,e}\mathbb{E}\left[ d(S^n,\hat{S}^n) | \mathcal{E}_d \right]\\
    &\overset{(a)}{\leq} (1-P_{d,e})\mathbb{E}_{\textbf{C}}\left[\sum_{s^n,\hat{s}^n}d(s^n,\hat{s}^n)\right.\\
    &\left. \sum_{\substack{m,i,\gamma,s^n_e,\\ x^n,y^n,z^n}}Q_{M I \Gamma S^n_e S^n   X^n Y^n Z^n \hat{S}^n}(m,i,\gamma,s^n_e, s^n,   x^n ,y^n ,z^n, \hat{s}^n)\right] \\
    &\quad\quad\quad\quad\quad\quad\quad\quad + P_{d,e} D_{max}\\
    &\leq  (1-P_{d,e})\sum_{s^n,\hat{s}^n}d(s^n,\hat{s}^n)\\
    &\quad \mathbb{E}_{\textbf{C}}\left[\sum_{\substack{m,i,\gamma,s^n_e,\\ x^n,y^n,z^n}}Q_{M I \Gamma S^n_e S^n   X^n Y^n Z^n \hat{S}^n}(m,i,\gamma,s^n_e, s^n,   x^n ,y^n ,z^n, \hat{s}^n)\right] \\
    &\quad\quad\quad\quad\quad\quad\quad\quad + P_{d,e} D_{max}\\
    &\overset{(b)}{=}(1-P_{d,e})\sum_{s^n,\hat{s}^n}d(s^n,\hat{s}^n)P^n_{S\hat{S}}(s^n,\hat{s}^n) + P_{d,e} D_{max} \\
    &\leq (1-P_{d,e})D + P_{d,e} D_{max},
\end{align*}
where $(a)$ follows by the fact that $\mathcal{E}_d^c$ implies the correct decoding, $(b)$ follows by the definition of $\mathcal{P}(D,P)$. To bound the true distortion (distortion under distribution $\hat{Q}$), we invoke Lemma \ref{lem: property of tv}\ref{def: tv property 1} again and it follows that
\begin{align*}
    &|\mathbb{E}_{\hat{Q}}[d(S^n,\hat{S}^n)] - \mathbb{E}_{Q}[d(S^n,\hat{S}^n)]|\\
     &\leq D_{max}|| Q_{S^n\hat{S}^n} - \hat{Q}_{S^n\hat{S}^n} ||\\
    &\leq D_{max}||Q_{M I \Gamma S^n_e S^n   X^n Y^n Z^n \hat{S}^n} -\bar{Q}_{M I \Gamma S^n_e S^n   X^n Y^n Z^n \hat{S}^n} || \\
    &\leq D_{max}\frac{3}{2}\exp\{-an\}.
\end{align*}
With sufficiently large block length $n$, we have
\begin{align*}
    \mathbb{E}_{\hat{Q}}[d(S^n,\hat{S}^n)] \leq D.
\end{align*}
So far, we have proved the achievability of region $C(D,\frac{3}{2}\exp\{-cn\})$. The proof is completed by using an optimal transportation argument as in \cite{saldi2015output}. Let $T_{\hat{S}^{'n}|\hat{S}^n}$ be a conditional distribution that converts the random sequence $\hat{S}^n$ to $\hat{S}^{'n}$. The argument shows that when the total variation of the distributions of $\hat{S}^n$ and $\hat{S}^{'n}$ is close,  one can construct a random sequence $\hat{S}^{'n}$ such that $P_{\hat{S}^{'n}}=P_S^n$ with an additional distortion that decays exponentially fast with the block length $n$, and the term $T_{\hat{S}^{'n}|\hat{S}^n}$ is called the optimal coupling. For completeness, we provide the argument in Appendix \ref{app: optimal transport argument}. This completes the achievability proof.

\section{proof of theorem \ref{the: cpd with cr and causal csi and strict constraint}}\label{app: proof of theorem de and strict constraint}
In this section, we prove the converse part of Theorem \ref{the: cpd with cr and causal csi and strict constraint}. Note that the reliable communication over the channel is equivalent to considering a channel $P_{YZ|XS_e}(y,z|x,s_e)=\sum_{s}P_{YZ|XS}(y,z|x,s)P_{S|S_e}(s|s_e)$. The bound on $R$ follows similarly to the converse proof of channel with causal side information\cite{shannon1958channels}.
\begin{align*}
    nR &= H(M)\\
    &\leq I(M;Y^n)\\
    &=\sum_{i=1}^n I(M;Y_i|Y^{i-1} )\\
    &\leq \sum_{i=1}^n I(M,\Gamma,S_e^{i-1},X^{i-1},Y^{i-1} ;Y_i)\\
    &\overset{(a)}{=}\sum_{i=1}^n I(M,\Gamma,S_e^{i-1},X^{i-1};Y_i)\\
    &\overset{(b)}{=}\sum_{i=1}^n I(M,\Gamma,S_e^{i-1};Y_i) \\ 
    &\overset{(c)}{=} n \cdot \frac{1}{n}\sum_{i=1}^n I(M,\Gamma,S_e^{Q-1};Y_Q|Q=i)\overset{(d)}{=}nI(U;Y)
\end{align*}
where $(a)$ follows by the Markov chain $Y_i-(S_{e}^{i-1},X^{i-1})-Y^{i-1}$, $(b)$ follows by the fact that $X^{i-1}$ is a deterministic function of $(M,\Gamma,S^{i-1}_e)$, $(c)$ follows by introducing a time-sharing random variable $Q\in[1,\dots,n]$, $d$ follows by setting $U_Q=(M,\Gamma,S_e^{Q-1}),U=(U_Q,Q),Y=Y_Q$. 

For the sum rate, it follows that
\begin{align*}
    n(R + R_c) &\geq H(M,\Gamma) \\
    &\geq I(M,\Gamma;S^n_e,\hat{S}^n)\\
    &=\sum_{i=1}^n I(M,\Gamma;S_{e,i},\hat{S}_i|S_{e}^{i-1},\hat{S}^{i-1})\\
    &\overset{(a)}{=}\sum_{i=1}^n I(M,\Gamma,S_{e}^{i-1},\hat{S}^{i-1};S_{e,i},\hat{S}_i)\\
    &\geq \sum_{i=1}^n I(M,\Gamma,S_{e}^{i-1};S_{e,i},\hat{S}_i)\\
    &=n \cdot \frac{1}{n} \sum_{i=1}^n I(U_Q;S_{e,Q},\hat{S}_Q|Q=i)\\
    &\overset{(b)}{=}n I(U_Q,Q;S_{e,Q},\hat{S}_Q)= nI(U;S_e,\hat{S}).
\end{align*}
where $(a)$ follows by the i.i.d. property of $(S_{e,i},\hat{S}_i),i=1,2,\dots$, $(b)$ follows by the independence between $Q$ and $(S_{e,Q},\hat{S}_Q)$.

\section{optimal transport argument}\label{app: optimal transport argument}
In this section we provide the explanation of the optimal transport argument in Appendix \ref{app: proof of thm noisy csi}, which is a direct application of the argument in \cite{saldi2015output} to our model. Given a cost function $c(S^n,\hat{S}^n)$ such that $c:S^n\times\hat{S}^n\to[0,+\infty)$ is a matric on $\mathcal{S}^n=\hat{\mathcal{S}}^n$, the optimal transportation is defined by
\begin{align}
    \label{def: optimal transport cost}\hat{T}_n(P_{\hat{S}^n},P_S^n)=\inf \{\mathbb{E}\left[ c(S^n,\hat{S}^n) \right]:S^n \sim P_S^n,\; \hat{S}^n \sim P_{\hat{S}^n}\},
\end{align}
where the infimum is taken over all joint distributions $P_{S^n\hat{S}^n}$ giving marginal distributions defined in \eqref{def: optimal transport cost}. The distribution giving the infimum in \eqref{def: optimal transport cost} is called the optimal coupling of $P_{\hat{S}^n}$ and $P_S^n$. We follow the convention in \cite{saldi2015output} and call the corresponding conditional distribution of $S^n$ given $\hat{S}^n$ an optimal coupling as well, denoted by $T_{\hat{S}^{'n}|\hat{S}^n}$. Note that given the distortion function one can further define
\begin{align*}
    \rho_n(S^n,\hat{S}^n)=(\sum_{i=1}^n d(S_i,\hat{S}_i)^p)^{\frac{1}{q}},
\end{align*}
 which is also a metric on $\mathcal{S}^n$. In our case, we have $p=q=1$. However, the arguments can be extended to more general positive $p=q\geq 1$. It follows that 
\begin{align*}
    &\hat{T}_n(P_{\hat{S}^n},P_S^n)\\
    &=\inf \{\mathbb{E}\left[ d(S^n,\hat{S}^n) \right]:S^n \sim P_S^n,\; \hat{S}^n \sim P_{\hat{S}^n}\}\\
    &=\frac{1}{n}\inf\left\{ \mathbb{E}\left[ \sum_{i=1}^n d(S_i,\hat{S}_i)^p \right]^{\frac{1}{q}}:S^n \sim P_S^n,\; \hat{S}^n \sim P_{\hat{S}^n} \right\}^q \\
    &=\frac{1}{n}(W_q(P_{\hat{S}^n},P_S^n))^q,
\end{align*}
where $W_q$ is the Wasserstein distance of order $q$ with cost function $\rho_n$. Using \cite[Theorem 6.15]{villani2009optimal} and \cite[Particular Case 6.16]{villani2009optimal} we have
\begin{align}
    W_q(P_{\hat{S}^n},P_S^n)&\leq 2^{\frac{1}{q'}}\left( \sum_{s^n} \rho_n^q(s^n_0,s^n) |P_{\hat{S}^n}(s^n)-P_S^n(s^n)|\right)^{\frac{1}{q}}\notag \\
    &=2^{\frac{1}{q'}}\left( \sum_{s^n} \sum_{i=1}^n d(s_{0,i},s_i)^p |P_{\hat{S}^n}(s^n)-P_S^n(s^n)|\right)^{\frac{1}{q}}\notag\\
    \label{ine: Wasserstein distance}&\leq D_{max}||P_{\hat{S}^n}-P_S^n||\leq D_{max}\frac{3}{2}\exp\{-cn\}.
\end{align}
for some $c>0$ and $q'$ such that $\frac{1}{q'}+\frac{1}{q}=1$. 
Now we apply the optimal coupling $T_{\hat{S}^{n'}|\hat{S}^n}$ to the reconstructed sequence $\hat{S}^n$. By the triangle inequality of a metric we can bound the distortion as
\begin{align*}
    &\left(\mathbb{E}\left[ (\frac{1}{n}\sum_{i=1}^n d(S_i,\hat{S}_i')^p) \right] \right)^{\frac{1}{q}}\\
    &\leq  \left(\mathbb{E}\left[ (\frac{1}{n}\sum_{i=1}^n d(S_i,\hat{S}_i)^p) \right] \right)^{\frac{1}{q}} + \left(\mathbb{E}\left[ (\frac{1}{n}\sum_{i=1}^n d(\hat{S}_i,\hat{S}_i')^p) \right] \right)^{\frac{1}{q}}\\
    &\overset{(a)}{=}\mathbb{E}[d^n(S^n,\hat{S}^n)] + \hat{T}_n(P_{\hat{S}^n},P_S^n) \leq D + \delta,
\end{align*}
where $(a)$ follows by settig $p=q=1$, $(b)$ follows by Appendix \ref{app: proof of thm noisy csi} and \eqref{ine: Wasserstein distance}. The proof is completed.

\section{proof of theorem \ref{the: cpddr with cr and noisy CSI and CSI}}\label{app: proof of henchman theorem}

    We first consider the system induced by the idealized distribution defined in \eqref{def: idealized distribution} in Appendix \ref{app: proof of thm noisy csi}.
The distribution $P_{Y^n\hat{S}^n|U^n}$ is the same as considering the distribution $P_{X^nS_e^nY^n\hat{S}^n|U^n}$ with additional side information $X^n,S^n_e$ at the encoder side. Since the decoder only observes the side information $Y^n$, considering $(U^n,Y^n,\hat{S}^n)$ is sufficient for the analysis.

We first argue for the best strategy for the henchman-eavesdropper. For a given henchman encoder $f_E$, suppose that the decoder is a stochastic decoder that maps each description $m_E\in[1:2^{nR_E}]$ to a sequence $\hat{S}_E^n\in\mathcal{S}^n$. Then, there exists a minimizer achieving a smaller distortion by choosing the $\hat{s}_E^n$ that minimizes the expectation of the distortion among all possible choices for each $m_E$, that is
\begin{align*}
    g(m_E,y^n) = \mathop{\arg\min}_{\hat{s}^n_E\in\mathcal{S}^n} \mathbb{E}\left[ d(\hat{S}^n,\hat{s}^n_E) \right].
\end{align*}
Hence, randomization does not help and the decoder of the eavesdropper is a deterministic decoder. Given an encoder $f_E$ such that $f_E(\hat{S}^n,Y^n)=M_E\in[1:2^{nR_E}]$ and side information $y^n\in\mathcal{Y}^n$, the decoder maps each $m_E$ to a unique reconstruction sequence with the help of side information $Y^n$. 
Without loss of generality, we can assume that there exists an interim communication codebook $\mathcal{C}_V=\{v^n\}$, and given $Y^n$, the encoder maps each $m_E$ to a unique codeword in this interim codebook such that the eavesdropper can recover the interim codeword $V^n$ based on $M_E$ and $Y^n$. The reconstructed sequence is then determined by $\hat{s}^n_E(M_E,Y^n)=v^n$. The optimal encoder in this case selects the codeword in the interim communication codebook that minimizes $\mathbb{E}[d(\hat{S},\hat{s}^n_E(M_E,Y^n))]$ and sends its description $M_E$ such that the decoder can locate the codeword with the help of the side information. 

The above argument turns the problem into a henchman problem for lossy communication with additional side information at the henchman and eavesdropper sides. Given a random communication codebook $\textbf{C}$, common randomness rate $R_c$ and a distortion value $D$, we now bound the distortion expectation $\mathbb{E}_{\textbf{C}}\left[ \mathbb{E}[d(\hat{S}^n,\hat{s}^n_E(M_E,Y^n))] \right]$. It follows that
\begin{align*}
    &\min_{f_E,\mathcal{C}_V}\mathbb{E}[d(\hat{S}^n,\hat{s}^n_E(M_E,Y^n))] \\
    &= \min_{f_E,\mathcal{C}_V} \mathbb{E}[d(\hat{S}^n,\hat{s}^n_E(M_E,Y^n))|d(\hat{S}^n,\hat{s}^n_E(M_E,Y^n)) > D] \\
    &\quad\quad\quad\quad\quad \cdot Pr\{d(\hat{S}^n,\hat{s}^n_E(M_E,Y^n)) > D\} \\
    &\quad\quad + \mathbb{E}[d(\hat{S}^n,\hat{s}^n_E(M_E,Y^n))|d(\hat{S},\hat{s}^n_E(M_E,Y^n)) \leq  D] \\
    &\quad\quad\quad\quad\quad \cdot Pr\{d(\hat{S}^n,\hat{s}^n_E(M_E,Y^n)) \leq D\}\\
    &\geq \min_{f_E,\mathcal{C}_V} D\cdot Pr\{d(\hat{S}^n,\hat{s}^n_E(M_E,Y^n)) > D\}.
\end{align*}
Note that $Pr\{d(\hat{S}^n,\hat{s}^n_E(M_E,Y^n)) > D\}$ is still a random variable due to the randomness of the communication codebook. It follows that
\begin{align*}
    &\mathbb{E}_{\textbf{C}}\left[ \min_{f_E,\mathcal{C}_V}\mathbb{E}[d(\hat{S}^n,\hat{s}^n_E(M_E,Y^n))] \right] \\
    &\geq D\cdot \mathbb{E}_{\textbf{C}}\left[ \min_{f_E,\mathcal{C}_V}  Pr\{d(\hat{S}^n,\hat{s}^n_E(M_E,Y^n)) > D\} \right] \\
    &\geq D \cdot (1-\tau)Pr_{\textbf{C}}\left\{\min_{f_E,\mathcal{C}_V}  Pr\{d(\hat{S}^n,\hat{s}^n_E(M_E,Y^n)) > D\} > (1-\tau) \right\}
\end{align*}
Hence, it suffices to show $Pr_{\textbf{C}}\left\{\max_{f_E,\mathcal{C}_V}  Pr\{d(\hat{S}^n,\hat{s}^n_E(M_E,Y^n)) \leq D\} > \tau_n \right\}\to 0$ for sufficiently large $n$ and small $\tau_n.$
Define an event $\mathcal{A}=\{(U^n(M,I,\Gamma),\hat{S}^n,Y^n)\in T^n_{P_{UY\hat{S}},\delta}\}$ and $\mathcal{B}=\{(U^n(M,I,\Gamma),Y^n)\in T^n_{P_{UY},\delta}\}$. We write $Pr\{\mathcal{A},\mathcal{B}\}$ as the probability that events $\mathcal{A}$ and $\mathcal{B}$ occur. Note that $Pr\{\mathcal{A}\}\to 1$ as well as $Pr\{\mathcal{B}\}\to 1$  as $n\to\infty$ since they are i.i.d. according to $P_{UY\hat{S}}^n$.  Then it follows that
\begin{align*}
    &\max_{f_E,\mathcal{C}_V}Pr\{d(\hat{S}^n,\hat{s}^n_E(M_E,Y^n)) \leq D,\mathcal{A},\mathcal{B}\}\\
    &=\max_{f_E,\mathcal{C}_V}\sum_{i,m,\gamma}\sum_{y^n}Pr\{d(\hat{S}^n,\hat{s}^n_E(M_E,Y^n))\leq D,U^n(M,I,\Gamma)\\
    &\quad\quad\quad\quad\quad\quad\quad\quad\quad=U^n(m,i,\gamma),Y^n=y^n ,\mathcal{A},\mathcal{B}\} \\
    &=\max_{f_E,\mathcal{C}_V}\sum_{y^n}\sum_{m,i,\gamma} Pr\{U^n(M,I,\Gamma)=U^n(m,i,\gamma),Y^n=y^n,\mathcal{B}\} \\
    &\quad\quad\quad\quad \cdot Pr\{d(\hat{S}^n,\hat{s}^n_E(M_E,Y^n)) \leq D,\mathcal{A}|\mathcal{B},U^n(m,i,\gamma),y^n\} \\
    &\overset{(a)}{\leq} \max_{\mathcal{C}_V,y^n}\sum_{m,i,\gamma} Pr\{U^n(m,i,\gamma)|y^n,\mathcal{B}\} \\
    &\quad \cdot Pr\{\min_{v^n\in \mathcal{V}^n}d(\hat{S}^n,\hat{s}^n_E(m_E(v^n),Y^n)) \leq D|U^n(m,i,\gamma),y^n,\mathcal{B}\} \\
    &\overset{(b)}{\leq} \max_{\mathcal{C}_V,y^n,v^n}2^{nR_E}\sum_{m,i,\gamma} Pr\{U^n(m,i,\gamma)\}\\
    &\quad\quad\quad\quad \cdot \mathbb{I}\{(U^n(m,i,\gamma),y^n)\in T^n_{P_{UY},\delta}\}\\
    &\quad\quad\quad\quad \cdot Pr\{d(\hat{S}^n,\hat{s}^n_E(m_E,Y^n)) \leq D|U^n(m,i,\gamma),y^n,\mathcal{B}\} \\
    &\overset{(c)}{\leq} \max_{v^n,y^n}2^{nR_E}\sum_{\gamma} 2^{-n(R+R_p+R_c-I(U;Y)-\delta)}\\
    &\quad\quad\quad\quad \cdot Pr\{d(\hat{S},\hat{s}^n_E(m_E,Y^n)) \leq D,\mathcal{A}|U^n(\gamma),y^n,\mathcal{B}\} \\
    &=\max_{v^n,y^n}2^{-n(R_c-R_E-2\delta)}\\
    &\quad\quad \cdot \sum_{\gamma} Pr\{d(\hat{S}^n,\hat{s}^n_E(m_E,Y^n)) \leq D,\mathcal{A}|U^n(\gamma),y^n,\mathcal{B}\}\\
\end{align*}
where $(a)$ follows by equivalently considering the case where the codeword $v^n$ minimizing the distortion is in the codebook $\mathcal{C}_V$ and is mapped to the lossy description $m_E$ by the optimal encoder $f_E$, the term $m_E(v^n)$ represents the encoding result when $v^n$ is selected. For simplicity, we omit $v^n$ and write $m_E(v^n)$ as $m_E$ in the following analysis. $(b)$ follows by the fact that the encoder selects codeword $v^n$ and sends index $m_E$ such that $v^n(m_E,y^n)$ is a deterministic function at the eavesdropper decoder side, which means given $y^n$, there are only $2^{nR_E}$ codewords in the codebook $\mathcal{C}_V$ that is possible to be picked by the encoder,  $(c)$ follows by the fact that the codeword $U^n(M,I,\Gamma)$ is uniformly selected from the codebook, and the two independently generated sequences $U^n(m,i,\gamma)$ and $y^n$ are jointly typical with probability lower bounded by $2^{-n(I(U;Y)+\delta)}$. Hence, there are at most $2^{n(R+R_p+R_c-I(U;Y)-\delta)}$ codewords in the codebook that are jointly typical with a given $y^n$, in $(c)$ we omit the indices $(m,i)$ since there are only $\gamma$ codewords is possible and the value of $(m,i)$ does not affect the probability. The remaining proof is similar to the argument in \cite[Sec. VIII]{schieler2016henchman} and we have $Pr_{\textbf{C}}\left\{\min_{f,\mathcal{C}_V}  Pr\{d(\hat{S},\hat{s}^n_E(M_E,Y^n)) \leq D\} > \tau_n \right\}\to 0$ if
\begin{align*}
    &R_E \leq R(D,P_{\hat{S}Y}),\\
    &R_E \leq R(D,P_{U\hat{S}Y}) + R_c.
\end{align*}

Finally, define event $\mathcal{E}_E$ as
\begin{align*}
    \mathcal{E}_E = \{d(\hat{S}^{'n},\hat{s}^n_E(M_E,Y^n))\geq \pi(R_E,R_c,P_{UY\hat{S}})-\epsilon\},
\end{align*}
where
\begin{align*}
    &\pi(R_E,R_c,P_{UY\hat{S}})=\\
    &\left\{
        \begin{aligned}
            &D(R_E,P_{\hat{S}Y})\quad\quad\quad\quad\quad\quad\quad\quad\quad\quad\quad\quad\;\; \text{if $R_E<R_c,$}\\
            &\min\{D(R_E,P_{\hat{S}Y}), D(R_E-R_c,P_{U\hat{S}Y})\}\quad\text{if $R_E\geq R_c.$}
        \end{aligned}
    \right.
\end{align*}
By taking the expectation over random codebook $\mathbf{C}$ and setting the length $n$ sufficiently large, we have $\mathbb{E}_{\textbf{C}}\left[ \min_{f_E,\mathcal{C}_V} Pr\{\mathcal{E}_E\}\right]\to 1$ and hence $\mathbb{E}_{\textbf{C}}\left[ \min_{f_E,\mathcal{C}_V}\mathbb{E}_Q[d(\hat{S}^n,\hat{S}^n_E)] \right]\geq D$ under the idealized distribution $Q$ defined in Eq. \eqref{def: idealized distribution} by setting $D=\pi(R_E,R_c,P_{UY\hat{S}})-\epsilon$. By Lemma \ref{lem: property of tv}\ref{def: tv property 1}, there exists a realization of the codebook such that $\mathbb{E}_{\textbf{C}}\left[ \min_{f_E,\mathcal{C}_V}\mathbb{E}_{\bar{Q}}[d(\hat{S}^n,\hat{S}^n_E)] \right]\geq D$ under the induced distribution $\bar{Q}$ defined in \eqref{def: induced distribution}, with the error probability and distortion constraint at the state estimator side being satisfied. It is proved that region $C(D,D_E,\epsilon)$ is achievable. To perfectly reconstruct the state distribution, the state estimator further uses an optimal transport $\hat{T}_{\hat{S}^{'n}|\hat{S}^n}$ such that $P_{\hat{S}^{'n}}=P_{S}^n$. Although the strategy of the henchman in the above argument only works for the case where the estimator reconstructs the state in a memoryless way, we always have the following inequality.
\begin{align*}
    \mathbb{E}\left[ d(\hat{S}^{n},\hat{S}^n_E) \right] &\leq \mathbb{E}\left[ d(\hat{S}^{'n},\hat{S}^n) \right] + \mathbb{E}\left[ d(\hat{S}^{'n},\hat{S}^n_E) \right]\\
    &\leq \hat{T}_n(P_{\hat{S}^n},P_S^n) + \mathbb{E}\left[ d(\hat{S}^{'n},\hat{S}^n_E) \right]\\
    &\leq 2^{-n\nu} + \mathbb{E}\left[ d(\hat{S}^{'n},\hat{S}^n_E) \right],
\end{align*}
where $\nu>0$, $\hat{T}_n(P_{\hat{S}^n},P_S^n)$ is the optimal transport cost. This is the case that we assume that the henchman observes the interim reconstructed sequence $\hat{S}^n$ and uses all of its bits to describe this interim sequence, which is obviously suboptimal compared to describing the final output sequence $\hat{S}^{'n}$. Thus, it follows that $\mathbb{E}\left[ d(\hat{S}^{'n},\hat{S}^n_E) \right] \geq \mathbb{E}\left[ d(\hat{S}^{n},\hat{S}^n_E) \right] - 2^{-n\nu}$ after the optimal transport and the proof is completed with $n\to\infty$.

\section{Converse part of corollary \ref{coro: deterministic cpddr with cr and noisy CSI and CSI}}\label{app: converse proof of corollary henchman deterministic channel}
In this section, we give the converse part of Corollary \ref{coro: deterministic cpddr with cr and noisy CSI and CSI}. The extended version of the type covering lemma is stated as follows.
\begin{lemma}{\cite[Lemma 4]{schieler2016henchman}}
    Let $\tau>0$ and $r>0$. Fix a joint type $P_{XY}$ on $\mathcal{X}^n\times\mathcal{Y}^n$, and let $y^n\in T_{Y}^n$. For $n>N$, where $N$ is a sufficiently large number related to $\tau$, there exists a codebook $\mathcal{C}(y^n,P_{XY})\subseteq \mathcal{Z}^n$ such that
    \begin{itemize}
        \item $\frac{1}{n}\log |\mathcal{C}(y^n,P_{XY})| \leq r,$
        \item For all $(x^n,y^n)\in T_{XY}^n$, 
        \begin{align*}
            \min_{z^n\in \mathcal{C}(y^n,P_{XY})}d(x^n,z^n) \leq D(r,P_{XY}) + \tau.
        \end{align*}
    \end{itemize}
\end{lemma}
By invoking the above lemma, the henchman observes the channel output $y^n$ and the joint type of $(\hat{s}^n,y^n)$. Since the number of types is upper bounded by $(n+1)^{|\hat{\mathcal{S}}||\mathcal{Y}|}$, the description of the types is negligible. The henchman finds the lossy description $z^n$ in the corresponding codebook and sends its index in the codebook and the index of the codebook (the joint type of $\hat{s}^n,y^n$). It follows that
\begin{align*}
    \mathbb{E}\left[d_E(\hat{S}^n,\hat{S}^n_E)\right] &\leq \mathbb{E}\left[ D(R_E+\epsilon,T_{\hat{S}^nY^n})\right] + \epsilon \\
    &\overset{(a)}{\leq} D(R_E+\epsilon,\mathbb{E}\left[T_{\hat{S}^nY^n}\right])\\
    &\overset{(b)}{=}D(R_E+\epsilon,P_{\hat{S}_QY_Q}) + \epsilon \\
    &\overset{(c)}{=}D(R_E+\epsilon,P_{\hat{S}Y}) + \epsilon,
\end{align*}
where $(a)$ is by the fact that $D(R_E+\epsilon,P_{\hat{S}Y})$ is concave in $P_{\hat{S}Y}$, and we defer the proof to the next paragraph,  $(\hat{S}_QY_Q)$ in $(b)$ are defined in the proof of Corollary \ref{coro: no common randomness capacity}, $(c)$ follows by setting $Y=Y_Q,\hat{S}=\hat{S}_Q$. 
To see the concavity of $D(R_E,P_{\hat{S}Y})$ (we omit $\epsilon$ here for simplicity), we define distribution $P_{\hat{S}Y}^{\lambda}=\lambda P_{\hat{S}Y}^{1}+(1-\lambda)P_{\hat{S}Y}^{2}$ for some $\lambda\in[0,1]$. Further let $I^{\lambda}(\hat{S};\hat{S}_E|Y)$ be mutual information computed according to the distribution $P_{\hat{S}Y}^{\lambda}$, and $I^{i}(\hat{S};\hat{S}_E|Y)$ be mutual information computed according to the distribution $P_{\hat{S}Y}^{i},i=1,2.$ Define $\mathbb{E}^{\lambda}[d(\hat{S}^n,\hat{S}^n_E)],\mathbb{E}^{i}[d(\hat{S}^n,\hat{S}^n_E)],i=1,2$ similarly. We have
\begin{align*}
    &D(R_E,P_{\hat{S}Y}^{\lambda}) \\
    &= \min_{\substack{P_{\hat{S}_E|\hat{S}Y}:\\I^{\lambda}(\hat{S};\hat{S}_E|Y)\leq R_E}} \mathbb{E}^{\lambda}[d(\hat{S}^n,\hat{S}^n_E)]\\
    &=\min_{\substack{P_{\hat{S}_E|\hat{S}Y}:\\I^{\lambda}(\hat{S};\hat{S}_E|Y)\leq R_E}} \lambda \mathbb{E}^{1}[d(\hat{S}^n,\hat{S}^n_E)] + (1-\lambda)\mathbb{E}^{2}[d(\hat{S}^n,\hat{S}^n_E)]\\
    &\overset{(d)}{\geq} \min_{\substack{P_{\hat{S}_E|\hat{S}Y}:\\ \lambda I^{1}(\hat{S};\hat{S}_E|Y) + \\(1-\lambda)I^{2}(\hat{S};\hat{S}_E|Y)\leq R_E}} \lambda \mathbb{E}^{1}[d(\hat{S}^n,\hat{S}^n_E)] + (1-\lambda)\mathbb{E}^{2}[d(\hat{S}^n,\hat{S}^n_E)]\\
    &\geq \min_{\substack{P_{\hat{S}_E|\hat{S}Y}:\\ \lambda I^{1}(\hat{S};\hat{S}_E|Y) + \\(1-\lambda)I^{2}(\hat{S};\hat{S}_E|Y)\leq R_E}} \lambda \mathbb{E}^{1}[d(\hat{S}^n,\hat{S}^n_E)] \\
    &\quad\quad\quad\quad\quad\quad\quad\quad + \min_{\substack{P_{\hat{S}_E|\hat{S}Y}:\\ \lambda I^{1}(\hat{S};\hat{S}_E|Y) + \\(1-\lambda)I^{2}(\hat{S};\hat{S}_E|Y)\leq R_E}} (1-\lambda) \mathbb{E}^{2}[d(\hat{S}^n,\hat{S}^n_E)]\\
    &\overset{(e)}{\geq}\min_{\substack{P_{\hat{S}_E|\hat{S}Y}:\\ I^{1}(\hat{S};\hat{S}_E|Y)\leq R_E \\I^{2}(\hat{S};\hat{S}_E|Y)\leq R_E}} \lambda \mathbb{E}^{1}[d(\hat{S}^n,\hat{S}^n_E)] \\
    &\quad\quad\quad\quad\quad\quad\quad\quad + \min_{\substack{P_{\hat{S}_E|\hat{S}Y}:\\ I^{1}(\hat{S};\hat{S}_E|Y) \leq R_E \\I^{2}(\hat{S};\hat{S}_E|Y)\leq R_E}} (1-\lambda) \mathbb{E}^{2}[d(\hat{S}^n,\hat{S}^n_E)]\\
    &\geq \min_{\substack{P_{\hat{S}_E|\hat{S}Y}:\\ I^{1}(\hat{S};\hat{S}_E|Y)\leq R_E }} \lambda \mathbb{E}^{1}[d(\hat{S}^n,\hat{S}^n_E)] \\
    &\quad\quad\quad\quad\quad\quad\quad\quad + \min_{\substack{P_{\hat{S}_E|\hat{S}Y}: \\I^{2}(\hat{S};\hat{S}_E|Y)\leq R_E}} (1-\lambda) \mathbb{E}^{2}[d(\hat{S}^n,\hat{S}^n_E)]\\
    &=\lambda D(R_E,P_{\hat{S}Y}^{1}) + (1-\lambda)D(R_E,P_{\hat{S}Y}^{2}).
\end{align*}
where $(d)$ follows by the fact that $I(\hat{S};\hat{S}_E|Y)$ is a concave function of $P_{\hat{S}Y}$, $(e)$ follows by the nonnegative property of mutual information.

To show that the equality in $(b)$ holds, let $p_{\hat{S}^nY^n}$ be the distribution coincide with the type of $(\hat{S}^n,Y^n)$. It follows that
\begin{align*}
    \mathbb{E}[p_{\hat{S}^nY^n}(\hat{s},y)] &= \sum_{\hat{s}^n,y^n}P_{\hat{S}^nY^n}(\hat{s}^n,y^n)p_{\hat{s}^ny^n}(\hat{s},y)\\
    &=\sum_{\hat{s}^n,y^n}P_{\hat{S}^nY^n}(\hat{s}^n,y^n) \frac{N(\hat{s},y|\hat{s}^n,y^n)}{n}\\
    &=\sum_{\hat{s}^n,y^n}P_{\hat{S}^nY^n}(\hat{s}^n,y^n) \frac{1}{n}\sum_{i=1}^n \mathbb{I}\{(\hat{s}_i,y_i)=(\hat{s},y)\}\\
    &=\frac{1}{n}\sum_{i=1}^n\sum_{\hat{s}_i,y_i}P_{\hat{S}_iY_i}(\hat{s}_i,y_i)\mathbb{I}\{(\hat{s}_i,y_i)=(\hat{s},y)\}\\
    &=\frac{1}{n}\sum_{i=1}^nP_{\hat{S}_iY_i}(\hat{s},y)\\
    &=\frac{1}{n}\sum_{i=1}^nP_{\hat{S}_QY_Q|Q}(\hat{s},y|Q=i)=P_{\hat{S}_QY_Q}(\hat{s},y).
\end{align*}
So far we have the region
\begin{equation*}
    C_{\epsilon}(D,D_E,0) = \left\{
    \begin{aligned}
        &(R, 0)\in\mathbb{R}^2: R\geq 0, \\
        &\exists P_{S_eS UXYZ\hat{S}} \in \mathcal{P},\\
        &R \leq H(Y|S_e) + \epsilon,\\
        &D \geq \mathbb{E}[d(S,\hat{S})],\\
        &D_E \leq D(R_E,P_{Y\hat{S}})+ \epsilon.
    \end{aligned}
    \right.
\end{equation*}
The argument of Eq. (96) in \cite{schieler2016henchman} is sufficient to show $\cap_{\epsilon>0}C_{\epsilon}(D,D_E,0) = C(D,D_E,0)$. This completes the proof of the converse part.

\end{document}